\documentclass[letterpaper, 10 pt, conference]{ieeeconf}  

\IEEEoverridecommandlockouts                              

\usepackage{amsmath} 
\usepackage{changepage}
\usepackage{subcaption}

\DeclareMathOperator*{\arginf}{arg\,inf}

\usepackage{amsthm}
\usepackage{graphicx}
\usepackage{xcolor}
\usepackage{cite}
\usepackage{lipsum}

\usepackage{amssymb}  
\title{\LARGE \bf
Minimum-norm Sparse Perturbations for Opacity in Linear Systems
}

\author{Varkey M. John and Vaibhav Katewa
\thanks{V. M. John is with the Robert Bosch Center for CPS (RBCCPS) at the Indian Institute of Science (IISc) Bangalore. V. Katewa is with the RBCCPS and Department
of ECE at IISc Bangalore. Email IDs:
        {\tt\small \{varkeym@iisc.ac.in,
vkatewa@iisc.ac.in\}}. This work is supported in part by SERB Grant SRG/2021/000292.} }

\newtheorem{theorem}{Theorem}

\newtheorem{lemma}{Lemma}
\usepackage{algorithm}
\theoremstyle{definition}
\newtheorem{definition}{Definition}
\newtheorem{example}{Example}
\theoremstyle{remark}
\newtheorem{remark}{Remark}
\theoremstyle{definition}

\newcommand\numberthis{\addtocounter{equation}{1}\tag{\theequation}}
\begin{document}
\maketitle
\thispagestyle{empty}
\pagestyle{empty}

\begin{abstract}
Opacity is a notion that describes an eavesdropper’s inability to estimate a system’s `secret' states by observing the system’s outputs. In this paper, we propose algorithms to compute the minimum sparse perturbation to be added to a system to make its initial states opaque. For these perturbations, we consider two sparsity constraints - structured and affine.
We develop an algorithm to compute the global minimum-norm perturbation for the structured case. For the affine case, we use the global minimum solution of the structured case as initial point to compute a local minimum. Empirically, this local minimum is very close to the global minimum.
We demonstrate our results via a running example.
\end{abstract}
\section{Introduction}
Privacy in Cyber-Physical Systems (CPS) has attracted significant interest in recent years, particularly due to increasing connectivity and computation capability in embedded devices. Recent works on privacy in CPS have explored ideas on opacity, differential privacy and information-theoretic security, among others \cite{DiffPriv_ACC_2016,InfPriv_ARC_2019,OpacityLinearSystems_IEEETAC_2020}.

Opacity, in particular, is a privacy notion which was introduced in the computer science literature, and has later been studied in depth for discrete event systems. In brief, opacity deals with the estimation of the `secret' states of a system from its outputs. In recent years, research has also been conducted to study this property's relevance in control systems \cite{OpacityLinearSystems_IEEETAC_2020,OpacityTrade-off_IEEECDC_2022}. These works have established connections between opacity and various properties of such systems. However, these works could only give information on whether the states of a system are opaque or not, and they do not provide information on how close a non-opaque system is to an opaque system. In this paper, we build on our previous work on the notion of opacity in linear systems \cite{OpacityTrade-off_IEEECDC_2022}, and develop algorithms that find the minimum-norm sparse perturbation required to make a system opaque.

In the literature, algorithms to compute a system's distance to uncontrollability, unobservability, instability, etc. have been studied for a long time (for example, see \cite{DistanceCont_IEEETAC_1986}). Recent works have also investigated the minimum perturbation required for the existence of invariant zeros. For instance, the authors in \cite{FastAlgo_IEEECDC_2008} developed a fast algorithm to compute the value of complex $s$ which gives the minimum perturbation for an invariant zero to exist in the system. Further, given a value of $s$, in \cite{Transmissionzerosradius_IFAC_2008}, the same authors developed closed form expressions to compute the associated minimum-norm real perturbation. However, none of these works have considered sparsity in the perturbations. 

Along the lines of structured perturbations, the authors in \cite{SubmatrixPert_IMANA_1988} computed the minimum perturbation required to be added to the lower-right sub-matrix $D$ such that the complete matrix loses rank. The authors in \cite{StrucDistPropP_TAC_2018} developed an algorithm to construct the minimum perturbation under sparsity constraints so that system will lose property $\mathcal{P}$. In a more recent work, the authors in \cite{RankRelaxation_arXiv_2022} proposed a rank-relaxation algorithm using truncated nuclear norm for sparse perturbations as applied to both operator 2-norm and Frobenius norm. In both \cite{StrucDistPropP_TAC_2018} and \cite{RankRelaxation_arXiv_2022}, the authors could compute only the local minimum and not the global minimum.

In our work, we develop algorithms to compute the global and local minimum-norm perturbations for the structured and affine sparsity constraints, respectively. Further, we improve the affine perturbation algorithm by using the solution to the structured perturbation case as the initial point.

\noindent {{\bf Notation:}} 
$\sigma_i(A,B)$ and $\sigma_i(A)$ denote the $i^\text{th}$ largest generalized singular value of pair $(A,B)$ and $i^\text{th}$ largest singular value of matrix $A$, respectively. $A^H$ denotes the conjugate-transpose of matrix $A$. diag$(d_1,d_2,\cdots,d_m)$ denotes an $m\times m$ diagonal matrix with diagonal elements ordered as $d_1,d_2,\cdots,d_m$. $A(i,j)$ denotes the element of matrix $A$ at $i^\text{th}$ row and $j^\text{th}$ column. $I$ denotes the identity matrix with appropriate dimension. $\mathcal{S}_1\backslash\mathcal{S}_2$ denotes the set difference operation. $\phi$ denotes the empty set. $j$ is the imaginary unit  $j=\sqrt{-1}$.

\section{Problem Formulation}
\subsection{Background}
We consider a discrete-time linear time-invariant system (denoted by $\Gamma$): 
{\small\begin{align*}
\Gamma\mathpunct{:}\quad\begin{aligned} x(k+1)&=Ax(k)+Bu(k),\\
y(k)&=Cx(k)+Du(k),\end{aligned}\numberthis\label{normal_model}
\end{align*}}

\noindent where $x\in \mathbb{R}^n, y\in\mathbb{R}^m, u\in \mathbb{R}^p, k\in \mathbb{Z}$ represent the state, output, input and time instant, respectively. Let 
{\small\begin{align*}
Y_{x(0),U(k)}=\begin{bmatrix}y(0)^T&y(1)^T&\cdots& y(k)^T\end{bmatrix}^T,
\end{align*}}

\noindent denote the output sequence (vector) produced by applying the input sequence 
{\small\begin{align*}
U(k)=\begin{bmatrix}u(0)^T&u(1)^T&\cdots& u(k)^T\end{bmatrix}^T,
\end{align*}}

\noindent to an initial state $x(0)\in\mathcal{X}_0$. We assume $\mathcal{X}_0=\mathbb{R}^n$ and $m\geq p$. If $m< p$, the analysis is performed on the system with transformed matrices $(A^T,C^T,B^T,D^T)$ (for which $m>p$).

We consider a set of secret initial states, denoted by $\mathcal{X}_s$ ($\mathcal{X}_s\subseteq \mathcal{X}_0)$, that a system operator wishes to keep private from eavesdroppers. The remaining set of non-secret initial states is $\mathcal{X}_{ns}= \mathcal{X}_0\backslash\mathcal{X}_s$. Any element of $\mathcal{X}_{ns}$ is not considered sensitive to disclosure. We use $x_s(0)$ and $x_{ns}(0)$ to denote individual elements in $\mathcal{X}_s$ and $\mathcal{X}_{ns}$, respectively.

\begin{definition}[\emph{Opacity of State \cite{OpacityTrade-off_IEEECDC_2022}}] \label{def:opacity_state} 
A secret initial state $x_s(0)\in \mathcal{X}_s$ is opaque with respect to the non-secret initial state set $\mathcal{X}_{ns},$ if, for all $k\geq 0$, the following property holds: for every $U_{s}(k)$, there exist $x_{ns}(0)\in \mathcal{X}_{ns}$ and $U_{ns}(k)$ such that $Y_{x_s(0),U_s(k)}=Y_{x_{ns}(0),U_{ns}(k)}$.\hfill $\square$
\end{definition}



\begin{definition}[\emph{Weakly Unobservable Subspace (WUS) \cite{ExtContObs_TAC_1976}} ]\label{def:WUS}
The weakly unobservable subspace of system \eqref{normal_model} (denoted by $\mathcal{V}(\Gamma)$) is defined as $\mathcal{V}(\Gamma) =\{x\in\mathbb{R}^n:\exists\: U(k)\text{ such that }Y_{x,U(k)}=0,\:\forall\:k\geq 0\}.$\hfill $\square$
\end{definition}


In our previous work \cite{OpacityTrade-off_IEEECDC_2022}, we showed that opaque states exist in a system if and only if the WUS is non-trivial (that is, $\mathcal{V}(\Gamma)\neq\{0\}$). We show next that this condition is also equivalent to the existence of invariant zeros in the system.
\begin{definition}[\emph{Invariant Zero}]\label{def:inv_zero}
    A number $s\in\mathbb{C}$ is an invariant zero of System $\Gamma$ in \eqref{normal_model} if:\\
    {\small\begin{align*}
    \text{rank}\Bigg(\underbrace{\begin{bmatrix}A-sI & B\\
     C & D\end{bmatrix}}_{\triangleq \Lambda_s}\Bigg)<n+p.\numberthis \label{eq:inv_zero}
    \end{align*}}\hfill $\square$
\end{definition}

\begin{lemma}\label{lem:ISOiffinv_zero}
    For System $\Gamma$ in \eqref{normal_model}, there exist opaque secret states  if and only if there exists an invariant zero in $\Gamma$.
\end{lemma}
\begin{proof}
    For System $\Gamma$ in \eqref{normal_model}, there exist opaque secret states if and only if $\mathcal{V}(\Gamma)\neq \{0\}$ \cite{OpacityTrade-off_IEEECDC_2022}. Further, $\mathcal{V}(\Gamma)\neq \{0\}$ if and only if there exists an invariant zero in $\Gamma$ \cite{Book_2002_ControlTheory}.
\end{proof}

\subsection{Problem Formulation}
There may exist some systems $(A,B,C,D)$ for which no opaque sets exist (that is, no invariant zeros exist). For such systems, a system operator wishes to change/perturb the system matrices $(A,B,C,D)$ such that the new system has opaque sets. Moreover, it wishes to to keep the perturbations ``small" for operational reasons. In this paper, we aim to find minimal system perturbations that make a non-opaque system opaque. We consider two problems with different types of perturbations - structured perturbations and arbitrary sparsity perturbations. We define these two problems in this subsection, and solve them in the subsequent sections.

\subsubsection{Problem 1}
In this problem, we consider perturbations with structured sparsity constraints. The perturbed system $\Gamma^P$  with matrices $(A_P,B_P,C_P,D_P)$ is given by 

{\small\begin{align*}
    \begin{bmatrix}A_P&B_P\\C_P&D_P\end{bmatrix}=\begin{bmatrix}A&B\\C&D\end{bmatrix}-E\Delta G,\numberthis\label{eq:pertform}
\end{align*}}

\noindent where $E$ and $G$ are diagonal matrices containing binary entries $\{0,1\}$. By selecting the appropriate  diagonal entries of $E$ and $G$, one can make corresponding rows and columns of $E\Delta G$ zero, and we call it as structured perturbation.



 With this perturbation structure, we consider the following optimization problem:
{\small\begin{equation}\label{eq:pertopt_prob1.1}
\begin{aligned}
\arginf\limits_{\Delta\in\mathbb{R}^{(n+m)\times (n+p)}}\quad \|\Delta\|&\\
    s.t.  \quad  \mathcal{V}(\Gamma^P)&\neq \{0\},
\end{aligned}
\end{equation}}

\noindent where $\|\cdot\|$ is the spectral norm. Let  $\Delta^*$ be the solution to the optimization problem \eqref{eq:pertopt_prob1.1}. By Lemma \ref{lem:ISOiffinv_zero} and Definition \ref{def:inv_zero}, we have that optimization problem \eqref{eq:pertopt_prob1.1} is equivalent to

{\small\begin{equation}\label{eq:pertopt_prob1.3}
\begin{aligned}
    \arginf\limits_{s\in\mathbb{C},\: \Delta\in\mathbb{R}^{(n+m)\times (n+p)}}&\quad \|\Delta\|\\\
    s.t. \quad  & \text{rank}(\Lambda_s -E\Delta G)<n+p.
\end{aligned}
\end{equation}}
\subsubsection{Problem 2}
In this problem, we consider perturbations with arbitrary sparsity constraints, that is, any arbitrary set of elements of the system matrices may be perturbed. Such perturbation corresponds to the class of affine perturbations in the literature. Note that structured perturbations considered in Problem 1 cannot model such sparsity constraints. For instance, for the case where only three elements $A(1,1)$, $A(1,3)$ and $A(3,1)$ can be perturbed, one cannot construct a corresponding $E$ and $G$. Hence, this is a more general class of perturbations which contains structured perturbations.
In this model, we have the following optimization problem:
{\small\begin{equation}\label{eq:pertopt_prob2.1}
\begin{aligned}
&\arginf\limits_{s\in\mathbb{C},\Delta\in\mathbb{R}^{(n+m)\times (n+p)}}&&\!\!\!\!\! \left\|\sum\limits_{i=1}^lE_i\Delta G_i\right\|\\\
    &\qquad\:s.t. &&  \!\!\!\!\!\text{rank}\left(\Lambda_s -\!\sum\limits_{i=1}^lE_i\Delta G_i\right)\!<\!n\!+\!p,
\end{aligned}
\end{equation}}

\noindent  where $E_i\in\mathbb{R}^{(n+m)\times (n+m)}$, $G_i\in\mathbb{R}^{(n+p)\times (n+p)}$ and $l$ is the number of $E_i\Delta G_i$ needed for the specific sparsity pattern.

In the following, we will show that the globally optimal solution of Problem 1 can be obtained, while only locally optimal solutions to Problem 2 can be obtained. Thus, there is need to solve these problems separately.

\section{Problem 1: Structured Perturbation}\label{sec:4.1}
We solve problem \eqref{eq:pertopt_prob1.3} in two steps. First, we fix $s$ and compute closed-form expressions for the optimal perturbation in problem \eqref{eq:pertopt_prob1.3} and its norm. We denote this optimal perturbation as $\Delta_s^*$. This is computed by extending the result in \cite{SubmatrixPert_IMANA_1988} (for real matrices) to complex matrices. Second, we find the optimal value of $s$ which gives the minimum $\|\Delta_s^*\|$ by searching through all $s\in\mathbb{C}$. This is done using the algorithms in \cite{FastAlgo_IEEECDC_2008}. The $\Delta_s^*$ corresponding to this optimal $s$ is then the required $\Delta^*$ that solves \eqref{eq:pertopt_prob1.3}.


Next, we describe the details of the two sub-problems corresponding to the above two steps.

\subsection{Sub-problem 1: Minimum-norm perturbation $\Delta_s^*$ for a given $s\in\mathbb{C}$}
Problem \eqref{eq:pertopt_prob1.3} with a given $s\in\mathbb{C}$ can be reformulated as: 

{\small\begin{equation}\label{eq:pertopt_prob1.6}
\begin{aligned}
&\arginf\limits_{x\in \mathbb{C}^{n+p}\backslash \{0\},\Delta\in\mathbb{R}^{(n+m)\times (n+p)}}&&\|\Delta\|\\\
    &\qquad\:s.t. &&  \left(\Lambda_s-E\Delta G\right)x=0,
\end{aligned}
\end{equation}}


The $\Delta$ that solves \eqref{eq:pertopt_prob1.6} is denotes as $\Delta_s^*$. Next, we aim to decompose constraint \eqref{eq:pertopt_prob1.6}. Let $\Delta^\alpha$ denote the sub-matrix formed by removing the all-zero rows of $E\Delta G$. Further, let $\Delta^r$ ($r$ denotes ``reduced") denote the sub-matrix formed by removing the all-zero rows and all-zero columns of $E\Delta G$. Note that there exists a matrix $J$ with binary entries $\{0,1\}$ such that $\Delta^\alpha = \Delta^r J$. Let $r_z$ and $c_z$ be the number of zero rows and zero columns, respectively, of $E\Delta G$. Further, based on the indices of the elements of $\Lambda_s$ that are to be perturbed, we partition $\Lambda_s$ row-wise into two separate matrices. The first matrix is denoted by $\Lambda_s^\alpha$, which is formed by those rows of $\Lambda_s$ whose elements can be perturbed. The second matrix is denoted by $\Lambda_s^\beta$, which is formed by the remaining rows of $\Lambda_s$. With these partitions of $\Lambda_s$ and $E\Delta G$, it is easy to see that constraint \eqref{eq:pertopt_prob1.6} can be decomposed row-wise into two constraints $(\Lambda_s^\alpha -\Delta^r J)x =0$ and $ \Lambda_s^\beta x =0$.

The above decomposition shows that \eqref{eq:pertopt_prob1.6} depends only on the reduced $\Delta^r$ and not on the full $\Delta$. Further, we have
{\small\begin{align*}
\|\Delta^r\| \overset{(a)}{=} \|E\Delta G\|\leq\|E\|\|\Delta\|\|G\|\overset{(b)}{=}\|\Delta\|,
\end{align*}}

\noindent where $(b)$ holds since $\|E\|=\|G\|=1$ because they are binary diagonal matrices, and we refer to Appendix \ref{app:normeq} for the proof of $(a)$. This implies that the cost $\|\Delta\|$ in problem \eqref{eq:pertopt_prob1.6} can be replaced with $\|\Delta^r\| $. Thus, \eqref{eq:pertopt_prob1.6} is reformulated as
{\small\begin{equation}\label{eq:pertopt_prob1.7}
\begin{aligned}
&\arginf\limits_{x\neq 0,\Delta^r\in\mathbb{R}^{(n+m-r_z)\times (n+p-c_z)}}&&\|\Delta^r\|\\\
    &\qquad\:s.t.  &&  &\left(\Lambda_s^\alpha -\Delta^r J\right)x =0,\\
    &&&&\Lambda_s^\beta x =0.
\end{aligned}
\end{equation}}

We describe this procedure in the running example below.
\begin{example}\label{ex:2}
Consider the discrete-time version of Example 1 of \cite{Transmissionzerosradius_IFAC_2008}. 
For this system, $\mathcal{V}(\Gamma)=\{0\}$. Let $A_o,B_o,C_o,D_o$ be the system matrices of the original system in \cite{Transmissionzerosradius_IFAC_2008}. Since for the above system, $m<p$, we consider the modified system with matrices $A=A_o^T, B=C_o^T, C=B_o^T, D=D_o^T$:
{\small\begin{equation}\label{sys:ex2_modfd}
\begin{aligned}
x(k+1)&= \begin{bmatrix}
0.74 & -0.12 & -0.38 \\
-0.69 & 1.62 & -0.21 \\
-2.08 & 0.63 & 0.14
\end{bmatrix}x(k)+\begin{bmatrix}
1.06 \\
0.71 \\
0.61
\end{bmatrix}u(k),\\
y(k)&=\begin{bmatrix}
-1.23 & 1.02 & -0.66 \\
-0.26 & 2.51 & 1.13
\end{bmatrix}x(k)+ \begin{bmatrix}
1.33 \\
-2.89
\end{bmatrix}u(k).
\end{aligned}
\end{equation}}

For this modified system, we have $m>p$. Consider the case where we perturb only elements  $A_0(1,1),A_0(1,3),A_0(3,1)$ and $A_0(3,3)$. This is same as perturbing the elements in same indices in $A$. Hence, we have the sparsity constraint matrices as $E=\text{diag}(1,0,1,0,0), G=\text{diag}(1,0,1,0)$.
For $s=0.8297+0.5583j$, we have
{\scriptsize\begin{align*}
 \Lambda_s &= \begin{bmatrix}0.0897-0.5583j &\!\!\! -0.12 &\!\!\! -0.38 &\!\!\! 1.06\\
 -0.69 &\!\!\! 0.7903-0.5583i &\!\!\! -0.21 &\!\!\! 0.71\\
 -2.08 &\!\!\! 0.63 &\!\!\! -0.6897-0.5583j &\!\!\! 0.61\\
 -1.23 &\!\!\! 1.02 &\!\!\! -0.66 &\!\!\! 1.33 \\
 -0.26 &\!\!\! 2.51 &\!\!\! 1.13 &\!\!\!\!\!\!\!\! -2.89\end{bmatrix}.
\end{align*}}
Since we perturb only rows 1 and 3, we have
{\small\begin{align*}
\Lambda_s^\alpha &= \begin{bmatrix}0.0897-0.5583j &\!\!\!\! -0.12 &\!\! -0.38 &\!\! 1.06\\
-2.08 & 0.63 &\!\! -0.6897-0.5583j &\!\! 0.61\end{bmatrix}\!,\\
\Lambda_s^\beta &= \begin{bmatrix}-0.69 & 0.7903-0.5583i & -0.21 & 0.71\\
-1.23 & 1.02 & -0.66 & 1.33 \\ -0.26 & 2.51 & 1.13 & -2.89\end{bmatrix}.
\end{align*}}

The perturbations $\Delta$, $\Delta^\alpha$ and $\Delta^r$ have the structure:
{\small\begin{align*}
\Delta &= \begin{bmatrix}\star & 0 & \star & 0\\
0 & 0 & 0 & 0\\
\star & 0 & \star & 0\\
0 & 0 & 0 & 0\\
0 & 0 & 0 & 0
\end{bmatrix},\quad
\Delta^\alpha=\begin{bmatrix}\star & 0 & \star & 0\\
\star & 0 & \star & 0
\end{bmatrix}, \Delta^r = \begin{bmatrix}\star & \star\\
\star & \star
\end{bmatrix},
\end{align*}}

\noindent where $\star\in\mathbb{R}$. Further, $J$ is given by:
{\small\begin{flalign*}
    &&J=\begin{bmatrix}
        1 & 0 & 0 & 0\\
        0 & 0 & 1 & 0
    \end{bmatrix}.&&\square
\end{flalign*}}
\end{example}

Next, we simplify Problem \eqref{eq:pertopt_prob1.7} further. 
We find the QR decomposition of $\Lambda_s^{\beta^H}$ as
    {\small\begin{align*}
        \Lambda_s^{\beta^H}=\begin{bmatrix}Q_1 & Q_2\end{bmatrix}\begin{bmatrix}R_1 \\ 0\end{bmatrix}.
    \end{align*}}

Since $\begin{bmatrix}Q_1 & Q_2\end{bmatrix}$ is an orthogonal matrix, we have $Q_1^HQ_2=0$. The columns of $Q_2$ form a basis for the nullspace of $\Lambda_s^\beta$. Thus, the solution to the second constraint of \eqref{eq:pertopt_prob1.7} is $x=Q_2 d$ (for some vector $d$), such that, $\Lambda_s^\beta x=R_1^HQ_1^HQ_2d=0$. Hence, we have that \eqref{eq:pertopt_prob1.7} reduces to
{\small\begin{equation}\label{eq:pertopt_prob1.8}
\begin{aligned}
&\arginf\limits_{d,\Delta^r\in\mathbb{R}^{(n+m-r_z)\times (n+p-c_z)}}&& \|\Delta^r\|\\\
    &\qquad\:s.t. &&  \Delta^r JQ_2d=\Lambda_s^\alpha Q_2d,
\end{aligned}
\end{equation}}

Let the solution to \eqref{eq:pertopt_prob1.8} be denoted by $\Delta_s^{*^r}$ and let $JQ_2$ be denoted by $Q_2^\alpha$. 
Note that the solution $\Delta_s^{*^r}$ 
is not guaranteed to exist in general. The following lemma provides the necessary and sufficient condition for its existence.
\begin{lemma}\label{lem:condtns_complexpertgivens}
There exists a solution to \eqref{eq:pertopt_prob1.8} if and only if $Q_2^\alpha\neq 0$ and rank$(\Lambda_s^\beta)<n+p$.
\begin{proof}
The proof follows from \cite{SubmatrixPert_IMANA_1988} and hence, is omitted.
\end{proof}
\end{lemma}

Next, we proceed to solve \eqref{eq:pertopt_prob1.8}. In order to compute the minimum-norm perturbation $\Delta_s^{*^r}$, we first compute $\|\Delta_s^{*^r}\|$. For this, we have the following lemma. For this lemma and subsequent results, we use the following notation:
{\small\begin{align*}
\Pi(\gamma,M)=\begin{bmatrix}Re(M)&-\gamma Im(M)\\\gamma^{-1} Im(M)&Re(M)\end{bmatrix}.
\end{align*}}

\begin{lemma}\label{lem:Min_struct_norm}
    Let the conditions in Lemma \ref{lem:condtns_complexpertgivens} hold. Further, let $\delta$ denote the min\{number of rows, number of columns\} of the matrix $\Lambda_s^\alpha$. We have for problem \eqref{eq:pertopt_prob1.8}:
{\small\begin{align*}
\|\Delta_s^{*^r}\| &=\sup\limits_{\gamma\in(0,1]}\sigma_{2\delta-1}\big(\Pi(\gamma,\Lambda_s^\alpha Q_2),\Pi(\gamma,Q_2^\alpha)\big).\numberthis\label{eq:pertopt_prob1.9}
\end{align*}}

\end{lemma}
\begin{proof}
By Theorem 2.3.1 of \cite{RealRobustnessRadii_PhDThesis_2011}, we have that for any two matrices $M\in\mathbb{C}^{q\times t}$ and $N\in\mathbb{C}^{r\times t}$, the following relation holds:
{\small\begin{align*}
&\inf\limits_{\Delta\in\mathbb{R}^{q\times r}}\{\|\Delta\|:\text{rank}(M-\Delta N)<\min(q,t)\}\\
&=\sup_{\gamma\in (0,1]}\sigma_{2\min(q,t)-1}\Bigg(\begin{bmatrix}Re(M)&\gamma Im(M)\\\gamma^{-1}Im(M)&Re(M)\end{bmatrix},\\
&\qquad\qquad\qquad\qquad\qquad\qquad\begin{bmatrix}Re(N)&\gamma Im(N)\\\gamma^{-1}Im(N)&Re(N)\end{bmatrix}\bigg).
\end{align*}}

If $M=\Lambda_s^\alpha Q_2$ and $N=Q_2^\alpha$, then the constraint in \eqref{eq:pertopt_prob1.8} is equivalent to rank$(M-\Delta_s^r N)<\min(q,t)=\delta$. Therefore, with the above relation, we get the infimum of $\|\Delta_s^r\|$ that solves \eqref{eq:pertopt_prob1.8} is as given in \eqref{eq:pertopt_prob1.9}. By definition, this infimum is equal to $\|\Delta_s^{*^r}\|$.
\end{proof}

With the value of $\|\Delta_s^{*^r}\|$ obtained from Lemma \ref{lem:Min_struct_norm}, $\Delta_s^{*^r}$ can be computed using the following theorem.
\begin{theorem}
Let the conditions in Lemma \ref{lem:condtns_complexpertgivens} hold. The following perturbation solves \eqref{eq:pertopt_prob1.8}:
{\small\begin{align*}
        \Delta_s^{*^r}=\begin{bmatrix}Re(\Lambda_s^\alpha Q_2d) & Im(\Lambda_s^\alpha Q_2d)\end{bmatrix}\begin{bmatrix}Re(Q_2^\alpha d) & Im(Q_2^\alpha d)\end{bmatrix}^\dag,\numberthis\label{eq:pertopt_prob1.sbp1.1}
    \end{align*}}

\noindent where $d$ is a complex vector that satisfies the hermitian-symmetric inequality:
{\small\begin{align*}
d^H\Big(\underbrace{\|\Delta_s^{*^r}\|^2Q_2^{\alpha^H} Q_2^\alpha-\Big(\Lambda_s^\alpha Q_2\Big)^H \Lambda_s^\alpha Q_2}_{\hat{H}}\Big)d\geq\\ \Big|d^T\Big(\underbrace{\|\Delta_s^{*^r}\|^2Q_2^{\alpha^T} Q_2^\alpha-\Big(\Lambda_s^\alpha Q_2\Big)^T \Lambda_s^\alpha Q_2}_{\hat{S}}\Big)d\Big|.\numberthis\label{eq:pertopt_prob1.10}
\end{align*}}
\end{theorem}
\begin{proof}
We first prove $\Delta_s^{*^r}$ in \eqref{eq:pertopt_prob1.sbp1.1} solves \eqref{eq:pertopt_prob1.8}. From Proposition 6.3.6 of \cite{GeomSpecValSets_PhDThesis_2003}, we have that for any two matrices $X\in\mathbb{C}^{r\times w}$ and $Y\in\mathbb{C}^{q\times w}$, if there exists $\Delta$ such that $\Delta X=Y$, then the choice $\Delta^*=\begin{bmatrix}Re(Y)&Im(Y)\end{bmatrix}\begin{bmatrix}Re(X)&Im(X)\end{bmatrix}^\dag$ satisfies 
\begin{align*}
\Delta^* X&=Y \:\text{and}\numberthis\label{eq:pertopt_prob1.11}\\ \|\Delta^*\|&=\inf\limits_{\Delta\in\mathbb{R}^{q\times r}}\{\|\Delta\|:\Delta X=Y\}.\numberthis\label{eq:pertopt_prob1.12}
\end{align*}

Since the conditions in Lemma \ref{lem:condtns_complexpertgivens} hold, we have that there exists $d\neq 0$ such that \eqref{eq:pertopt_prob1.8} holds. Therefore, with the above result, we have that for $X=Q_2^\alpha d$ and $Y=\Lambda_s^\alpha Q_2 d$, it holds that $\Delta_s^{*^r}$ in \eqref{eq:pertopt_prob1.sbp1.1} satisfies \eqref{eq:pertopt_prob1.12}, and thus solves \eqref{eq:pertopt_prob1.8}.\par 
Next we prove that $d$ satisfies \eqref{eq:pertopt_prob1.10}. By Proposition 6.4.3 of \cite{GeomSpecValSets_PhDThesis_2003}, we have for all $\xi\in\mathbb{C}^w$ the following condition is satisfied by $X,Y$ and $\|\Delta^*\|$:
\begin{align*}
\xi^H(\|\Delta^*\|^2X^HX-Y^HY)\xi\geq|\xi^T(\|\Delta^*\|^2X^TX-Y^TY)\xi|\numberthis\label{eq:pertopt_prob1.13}
\end{align*}

Since $w=1$ for our choice of $X$ and $Y$, we have that $\xi\in\mathbb{C}$, and consequently, $\xi^H\xi=|\xi|^2$ on the left hand side gets cancelled with $|\xi^T||\xi|=|\xi|^2$ on the right, in the inequality in \eqref{eq:pertopt_prob1.13}. Let $\Delta^*=\Delta_s^{*^r}$. Hence, the condition \eqref{eq:pertopt_prob1.13} reduces to:
\begin{align*}
&\|\Delta_s^{*^r}\|^2(Q_2^\alpha d)^H(Q_2^\alpha d)-(\Lambda_s^\alpha Q_2 d)^H(\Lambda_s^\alpha Q_2 d)\geq\\
&\qquad\qquad\qquad\left|\|\Delta_s^{*^r}\|^2(Q_2^\alpha d)^T(Q_2^\alpha d)-(\Lambda_s^\alpha Q_2 d)^T(\Lambda_s^\alpha Q_2 d)\right|\\
&\iff d^H\left(\|\Delta_s^{*^r}\|^2Q_2^{\alpha^H} Q_2^\alpha-\left(\Lambda_s^\alpha Q_2\right)^H \Lambda_s^\alpha Q_2\right)d\geq\\ &\qquad\qquad\qquad\left|d^T\left(\|\Delta_s^{*^r}\|^2Q_2^{\alpha^T} Q_2^\alpha-\left(\Lambda_s^\alpha Q_2\right)^T \Lambda_s^\alpha Q_2\right)d\right|,
\end{align*}

\noindent which is the required condition that $d$ must satisfy, as given by \eqref{eq:pertopt_prob1.10}.
\end{proof}

All vectors $d$ that satisfy \eqref{eq:pertopt_prob1.10} forms a subspace $\mathcal{D}$. Further, any $d\in\mathcal{D}$ results in the same $\Delta_s^{*^r}$ in \eqref{eq:pertopt_prob1.sbp1.1}. An algorithm to compute a basis for $\mathcal{D}$ when $E$, $G$ are identity is given in \cite{DFMRadius_IFAC_2008}. This algorithm is modified to compute $d$ for the structured case as follows.

\begin{lemma}\label{lem:dbasis_pert}
Using the algorithm presented in Section 3.1 of \cite{DFMRadius_IFAC_2008}, a basis for $\mathcal{D}$ can be computed by choosing $\hat{H}=\|\Delta_s^{*^r}\|^2Q_2^{\alpha^H} Q_2^\alpha-\left(\Lambda_s^\alpha Q_2\right)^H \Lambda_s^\alpha Q_2$ and $\hat{S}=\|\Delta_s^{*^r}\|^2Q_2^{\alpha^T} Q_2^\alpha-\left(\Lambda_s^\alpha Q_2\right)^T \Lambda_s^\alpha Q_2$ (as given in \cite{DFMRadius_IFAC_2008}).
\end{lemma}
\begin{proof}
By Proposition 6.4.3 of \cite{GeomSpecValSets_PhDThesis_2003}, we require $d$ to belong to a subspace that satisfies \eqref{eq:pertopt_prob1.10} so that $\|\Delta_s^r\|$ is the minimum given by \eqref{eq:pertopt_prob1.9}. When $s$ is real, this choice of $\hat{H}$ and $\hat{S}$ is real and symmetric, and $\hat{H}=\hat{S}$. When $s$ is not real, $\hat{H}^H=\hat{H}$ and $\hat{S}^T=\hat{S}$. Hence, we have that $\hat{H}$ and $\hat{S}$ satisfy the conditions in Section 3.1 of \cite{DFMRadius_IFAC_2008}. Therefore, there exists a matrix $P$ that block-diagonalizes $\hat{H}$ and $\hat{S}$, according to the algorithms in \cite{DFMRadius_IFAC_2008}. Consequently, the columns of $P$  may be extracted to compute $d$ that satisfies \eqref{eq:pertopt_prob1.10}.
\end{proof}

Finally, $\Delta_s^*$ can be obtained by expanding $\Delta_s^{*^r}$ by inserting zero-rows and zero-columns in appropriate places, such that $E\Delta_s^* G=\Delta_s^*$. This $\Delta_s^*$ is one of the solutions to \eqref{eq:pertopt_prob1.6}, and other solutions may exist. However, the norms of all the solutions are equal. 

\textbf{Example \ref{ex:2} (Contd.)} We take the QR decomposition of $\Lambda_s^{\beta^H}$ to obtain $Q_2^\alpha$ as:
{\small\begin{align*}
Q_2^\alpha = \begin{bmatrix}
-0.0441 - 0.3712j \\
0.7653 + 0.3568j
\end{bmatrix}.
\end{align*}}

With Lemma \ref{lem:dbasis_pert}, we compute a basis for $\mathcal{D}$, which we then use to compute the required $\Delta^r$ for the modified system (rounded to 4 decimal points) from \eqref{eq:pertopt_prob1.sbp1.1}:
{\small\begin{align*}
d &= 3.7905-11.038j,\\
\Delta^r &= \begin{bmatrix}-0.0341 & -0.2048\\
0.0682 & -0.0307\end{bmatrix}.
\end{align*}}

It can be seen that $\|\Delta^r\|=0.2086$, which equals $\|\Delta_s^*\|$. Due to the perturbation, only $A$ gets perturbed, such that 
{\small\begin{align*}
A_P=\begin{bmatrix}0.7741 & -0.12 & -0.1752 \\
-0.69 & 1.62 & -0.21 \\
-2.1482 & 0.63 & 0.1707\end{bmatrix}.
\end{align*}}
The perturbed system has invariant zeros at $s=0.8297\pm 0.5583j$. Consequently, we have $\mathcal{V}(\Gamma^P)=\text{span}(I_3)\neq\{0\}$, thus permitting opaque sets.\hfill$\square$

\subsection{Sub-problem 2: $s\in\mathbb{C}$ with Infimum $\|\Delta_s^*\|$}\label{sec:s_forpert}
In this sub-problem, we find the $s$ for which the corresponding $\Delta_s^*$ has the minimum norm over all $s\in\mathbb{C}$, that is
\begin{equation}\label{eq:pertopt_prob1.14}
\begin{aligned}
    \arginf\limits_{s\in\mathbb{C}}&\quad \|\Delta_s^*\|.
\end{aligned}
\end{equation}

To solve \eqref{eq:pertopt_prob1.14} for unstructured perturbations (that is, for $E=I_{n+m},G=I_{n+p}$), algorithms that find the global minimum have been developed in the literature\cite{FastAlgo_IEEECDC_2008}. However, as per our knowledge, such algorithms for the structured case (that is, for $E$ and $G$ are not both identity matrices) have not been developed. We first summarize Algorithm 4.1 of \cite{FastAlgo_IEEECDC_2008}, which finds the minimum norm for unstructured perturbation. We then modify this algorithm to find the minimum norm of structured perturbations.

\paragraph{Algorithm for Unstructured Perturbation (Algorithm 4.1 of \cite{FastAlgo_IEEECDC_2008})}
An important constituent of the algorithm is a line search method that finds the minimum of a function, say $f(x)$, over the positive real line $\mathbb{R}^{+}$. The method is initialized with a (possibly disjoint) region $R_0$. In iteration $k$, a random point $x_k$ is picked from $R_{k-1}$, and the region $R_k$ is computed as:
{\small\begin{align*}
R_k = \{x: f(x) < f(x_k) \}. 
\end{align*}}
To compute $R_k$, the points $\{x:f(x)= f(x_k)\}$ are computed using Theorem 3.2 of \cite{FastAlgo_IEEECDC_2008}. This method ensures that $R_k$ shrinks as iterations progress and eventually converges to the minimum.

The algorithm approximates the search on the whole complex plane (in \eqref{eq:pertopt_prob1.14}) by a search on $L$ half lines in the first and second quadrants of the complex plane. These half lines are specified by their angles which form the set:
{\small\begin{align*}
\Theta_0 = \{\theta_1 =0 < \theta_2 < \theta_3 < \cdots < \theta_{L-1} < \pi = \theta_L\}.
\end{align*}}

Let $R_k^{\theta}$ denote the region corresponding to the line at angle $\theta$ at iteration $k$. The algorithm first finds (say at iteration $k_1$) the minimum value of $\|\Delta_{s}^*\|$ on the $\theta_1=0$ line. It then uses this minimum to initialize the region $R_{k_1+1}^{\theta_2}$ on the $\theta_2$ line, and finds the minimum on the $\theta_2$ line. In this way, it sweeps through all lines and finds the minimum over all lines.

Next, we build on the above algorithm to compute the minimum norm of structured perturbations.
\paragraph{Algorithm for Structured Perturbation}
The solution to the optimization problem \eqref{eq:pertopt_prob1.14} requires $s\in\mathbb{C}$ to satisfy the conditions in Lemma \ref{lem:condtns_complexpertgivens}. We assume that Lemma \ref{lem:condtns_complexpertgivens} is satisfied for all values of $s\in\mathbb{C}$. The case where this does not hold is elaborated in Remark \ref{rem:lem2_finites}.
We solve \eqref{eq:pertopt_prob1.14} for structured perturbations by modifying Algorithm 4.1 of \cite{FastAlgo_IEEECDC_2008} using \eqref{eq:pertopt_prob1.9}. In particular, the following three modifications are made:\\
(i) The values of $\|\Delta_{s_k}^*\|$ and $R_k^\theta$ are computed for the structured perturbation $\|\Delta^r\|$ (found using \eqref{eq:pertopt_prob1.9}).\\ 
(ii) The endpoints of $R_k^\theta$ (Step 3a of Algorithm \ref{alg:IV.2.2}) are determined using an approximation on $\sigma_{2\delta-1}\big(\Pi(\gamma_{k-1},\Lambda_{|s|e^{j\theta}}^\alpha Q_2),\Pi(\gamma_{k-1},Q_2^\alpha)\big)$. The details of the approximation and its accuracy are found in Appendix \ref{app:approxanalysis}, and computation of $R_k^\theta$ is found below.\\
(iii) We construct $\Theta_0$ by taking discrete values in $[0,2\pi]$, as against $[0,\pi]$ used in the unstructured case (for which symmetry was proven).\par
The endpoints of $R_k^\theta$, that is, the set of $|s|$ for a given $\theta$ in Step 3a of Algorithm \ref{alg:IV.2.2}, are computed by using Theorem 3.2 of \cite{FastAlgo_IEEECDC_2008} on an approximation of $\sigma_{2\delta-1}\big(\Pi(\gamma_{k-1},\Lambda_{|s|e^{j\theta}}^\alpha Q_2)$, $\Pi(\gamma_{k-1},Q_2^\alpha)\big)$.
This approximation is found as follows. As shown previously, the infimum norm of the solution to \eqref{eq:pertopt_prob1.3} is given by \eqref{eq:pertopt_prob1.9}. Therefore, we find an approximate solution to \eqref{eq:pertopt_prob1.3}. This is done using the approach shown in \cite{RealRobustnessRadii_PhDThesis_2011} as follows. First we convert $E$ and $G$ into non-singular square matrices $\breve{E}$ and $\breve{G}$, respectively. This is done by appending zero rows and zero columns and then adding a small value $\epsilon$ to the zero singular values of $E$ and $G$, so that these matrices become invertible. It is shown in \cite{RealRobustnessRadii_PhDThesis_2011} that for $\breve{E}$ we have $\|\breve{E}-E\|=\epsilon$ and the condition number of $\breve{E}$ is given by $\sigma_1(\breve{E})/\epsilon$ (similar relations hold for $\breve{G}$). 

With the approximation described previously, we note that the infimum norm of the solution to \eqref{eq:pertopt_prob1.3} is approximately \footnote{Empirically very close for the real matrices $E$ and $G$ (for instance, refer to \cite{RealRobustnessRadii_PhDThesis_2011} which approximates only $E$). Further, since $E$ and $G$ are diagonal matrices, $\breve{E}$ and $\breve{G}$ are invertible for any value of $\epsilon$. The numerical precision of the approximation depends on the software used. Our analysis on this approximation has mainly been empirical and we are yet to prove theoretical bounds on its accuracy. Also refer to the Appendix \ref{app:approxanalysis} for some preliminary analysis.}{equal} to the solution to the following optimization problem:
{\small\begin{equation}\label{eq:pertopt_prob_set_1}
\begin{aligned}
&\inf\limits_{s\in\mathbb{C},\:\breve{\Delta}\in\mathbb{R}^{(n+m)\times (n+p)}}&&\|\breve{\Delta}\|\\\
    &\qquad\:\text{s.t.} &&  \text{rank}\left(\breve{E}^{-1}\breve{\Lambda}_s\breve{G}^{-1}-\breve{\Delta}\right)<n+p.
\end{aligned}
\end{equation}}

Let $\breve{\Delta}_s^*$ be the solution to \eqref{eq:pertopt_prob_set_1}. We have \cite{Transmissionzerosradius_IFAC_2008,RealRobustnessRadii_PhDThesis_2011}:
{\small\begin{align*}
    \|\breve{\Delta}_s^*\|=\sup\limits_{\gamma\in (0,1]}\sigma_{2(n+p)-1}\big(\Pi(\gamma,\breve{E}^{-1}\breve{\Lambda}_s\breve{G}^{-1})\big)\numberthis\label{eq:approx_norm}
\end{align*}}

With this, we obtain the approximate values of the set of $|s|$ by finding real, non-negative eigenvalues of $\mathcal{H}(\|\Delta_{s_{k-1}}^*\|,\hat{A},\hat{B},\hat{C},\hat{D},\hat{\Delta},Y_{\gamma_{k-1,n+m}}^{-1}\widehat{\breve{E}^{-1}}P_{nm},P_{np}\widehat{\breve{G}^{-1}}Y_{\gamma_{k-1},n+p})$ for which the singular value $\sigma_{2(n+p)-1}\big(\Pi(\gamma_{k-1},\breve{E}^{-1}\breve{\Lambda}_{|s|e^{j\theta}}\breve{G}^{-1})\big)=\|\Delta_{s_{k-1}}^*\|$, where
{\small\begin{align*}
    \hat{A}&=\begin{bmatrix}
        A&0\\0&A
    \end{bmatrix},
    \hat{B}=\begin{bmatrix}
        B&0\\0&B
    \end{bmatrix},\\
    \hat{C}&=\begin{bmatrix}
        C&0\\0&C
    \end{bmatrix},
    \hat{D}=\begin{bmatrix}
        D&0\\0&D
    \end{bmatrix},
    \hat{\Delta}=\begin{bmatrix}
        e^{i\theta}I_{n}&0\\0&e^{-i\theta}I_{n}
    \end{bmatrix}\\
    \widehat{{\breve{E}^{-1}}}&=\begin{bmatrix}
        \breve{E}^{-1}&0\\0&\breve{E}^{-1}
    \end{bmatrix},
    \widehat{{\breve{G}^{-1}}}=\begin{bmatrix}
        \breve{G}^{-1}&0\\0&\breve{G}^{-1}
    \end{bmatrix},\\
    P_{nm}&=\begin{bmatrix}
        I_{n}&0&0&0\\
        0&0&I_{m}&0\\
        0&I_n&0&0\\
        0&0&0&I_m
    \end{bmatrix},
    P_{np}=\begin{bmatrix}
        I_{n}&0&0&0\\
        0&0&I_{n}&0\\
        0&I_p&0&0\\
        0&0&0&I_p
    \end{bmatrix},\\
    Y_{\gamma,i}&\triangleq \frac{1}{2\gamma}\begin{bmatrix}
        \sqrt{1+\gamma^2}I_i&j\gamma\sqrt{1+\gamma^2}I_i\\
        \sqrt{1+\gamma^2}I_i&-j\gamma\sqrt{1+\gamma^2}I_i
    \end{bmatrix},
\end{align*}}

\noindent and the function $\mathcal{H}$ is as defined in \cite{FastAlgo_IEEECDC_2008}. The above procedure to find the approximate set of $|s|$ is based on applying Theorem 3.2 of \cite{FastAlgo_IEEECDC_2008} on the matrices of \eqref{eq:pertopt_prob_set_1}. Note that this set is very close to the actual set. Further, the actual set may be found by taking trial points in the vicinity of the approximate set that satisfy $\sigma_{2\delta-1}\big(\Pi(\gamma_{k-1},\Lambda_{|s|e^{j\theta}}^\alpha Q_2),\Pi(\gamma_{k-1},Q_2^\alpha)\big)=\|\Delta_{s_{k-1}}^*\|$ (Step 3a of Algorithm \ref{alg:IV.2.2}), or by finding the roots numerically with initial points taken as $|s|e^{j\theta}$. 

The remainder of the algorithm remains the same as in \cite{FastAlgo_IEEECDC_2008}, and is described in Algorithm \ref{alg:IV.2.2}.
To summarize, Algorithm \ref{alg:IV.2.2} solves \eqref{eq:pertopt_prob1.14}, and with the optimal value of $s$, we can solve \eqref{eq:pertopt_prob1.3} (and hence, \eqref{eq:pertopt_prob1.1}) by using \eqref{eq:pertopt_prob1.sbp1.1}.

\textbf{Example \ref{ex:2} (Contd.)} Using Algorithm \ref{alg:IV.2.2}, we obtain the optimal values of $s\in\mathbb{C}$ that minimizes $\|\Delta_s^*\|$ as $0.8297+0.5583j$ for which $\|\Delta_s^*\|=0.2086$. As expected, this norm is greater than that in \cite{Transmissionzerosradius_IFAC_2008}, since we are considering structured perturbations. Figure \ref{fig:Example_Pert} shows a surface plot of $\|\Delta_s^*\|$, corresponding to different values of $s$ near the origin. It is evident from this plot that $\|\Delta^r\|$ (and hence $\|\Delta_s^*\|$) is minimum at $s=0.8297+0.5583j$.\hfill$\square$

\begin{algorithm}
\caption[Minimum Perturbation Norm (Structured)]{Minimum Perturbation Norm (Structured)}\label{alg:IV.2.2}
\textbf{Input:} Matrices $A$, $B$, $C$, $D$, $E$, $G$.\\
\textbf{Output:} $s$ and $\|\Delta_{s}^*\|$ which solves \eqref{eq:pertopt_prob1.14}.\\~\\
\textbf{Step 1:} Set $\Theta_0=\{\theta_1,\theta_2,\cdots,\theta_{L}\}$ and $k=0$. Choose arbitrary $s_0\in\mathbb{C}$. Compute $\|\Delta_{s_0}^*\|$ and $\gamma_0$ such that $\|\Delta_{s_0}^*\|=\sigma_{2\delta-1}\big(\Pi(\gamma_0,\Lambda_{s_0}^\alpha Q_2),\Pi(\gamma_0,Q_2^\alpha)\big)$.\\
\textbf{Step 2:} Update $k$ to $k+1$.\\
\textbf{Step 3a:} For each $\theta\in \Theta_{k-1}$, find set of $|s|$ that satisfies $\sigma_{2\delta-1}\big(\Pi(\gamma_{k-1},\Lambda_{|s|e^{j\theta}}^\alpha Q_2),\Pi(\gamma_{k-1},Q_2^\alpha)\big)=\|\Delta_{s_{k-1}}^*\|$ (using Theorem 3.2 of \cite{FastAlgo_IEEECDC_2008} on the approximation).\\
\textbf{Step 3b:} For each $\theta\in \Theta_{k-1}$, let $R_k^\theta$ denote those regions of $s\in\mathbb{C}$ that satisfy $\|\Delta_{s}^*\|<\|\Delta_{s_{k-1}}^*\|$. The set of $|s|$ in Step 3a forms the endpoints of the intervals that form $R_k^\theta$. Find $R_k^\theta$ using trial points in between these values of $|s|$.\\
\textbf{Step 4:} Update $\Theta_k$ by removing those $\theta$ for which $R_k^\theta=\phi$.\\
\textbf{Step 5:} Choose any $|s_k|\in R_k^\theta$ and corresponding $\theta$. Set $s_k=|s_k|e^{j\theta}$. (Note that since $|s_k|\in R_k^\theta$, $\|\Delta_{s_k}^*\|<\|\Delta_{s_{k-1}}^*\|$).\\
\textbf{Step 6:} From  $\|\Delta_{s_k}^*\|$, update value of $\gamma_k$.\\
\textbf{Step 7:} Stop if tolerance levels are violated. Else, go to Step 2.
\end{algorithm}

\begin{remark}\label{rem:lem2_finites}
In the above, we assumed that Lemma \ref{lem:condtns_complexpertgivens} is satisfied for all values of $s\in\mathbb{C}$. Let us consider the case where it is satisfied for finite $s\in\mathbb{C}$. Hence, we have the following two cases:\\
(i) $Q_2^\alpha=0$ for all but finite values of $s\in\mathbb{C}$. In this case, we can find the values of $s$ with the following steps:
First, let $x$ be split up row-wise into $x^\alpha=Jx$ and $x^\beta$. Next, we split $\Lambda_s^{\beta}$ column-wise into two matrices $\Lambda_s^{\beta,\alpha}$ and $\Lambda_s^{\beta,\beta}$, such that  the column indices of these matrices in $\Lambda_s^\beta$ are equal to the row indices of $x^\alpha$ and $x^\beta$, respectively, in $x$. From the second constraint of \eqref{eq:pertopt_prob1.7}, we have 
{\small\begin{align*}
\Lambda_s^{\beta} x=0\iff\Lambda_s^{\beta,\alpha} x^\alpha&=0,\\
\qquad\qquad\qquad\qquad\qquad\Lambda_s^{\beta,\beta} x^\beta&=0.
\end{align*}}

From the first equation above, we have
{\small\begin{align*}
\Lambda_s^{\beta,\alpha} x^\alpha=\Lambda_s^{\beta,\alpha} Jx=\Lambda_s^{\beta,\alpha} JQ_2d=\Lambda_s^{\beta,\alpha} Q_2^\alpha d=0.
\end{align*}}

Since $Q_2^\alpha = 0 \iff Q_2^\alpha d = 0\: \forall d$, we have that the above is equivalent to the fact that the nullspace of $\Lambda_s^{\beta,\alpha}$ is $\{0\}$, that is, $\Lambda_s^{\beta,\alpha}$ has  full column rank. Therefore, we find those values of $s\in\mathbb{C}$ that reduce the column rank of $\Lambda_s^{\beta,\alpha}$. This is done by performing column-exchanges to convert $\Lambda_s^{\beta,\alpha}$ in form of  matrix in \eqref{eq:inv_zero}, and then finding the invariant zeros of the corresponding system.\\
(ii) rank($\Lambda_s^\beta$) $=n+p$ for all but finite values of $s\in\mathbb{C}$.  
To find this set of $s$, column-exchanges are performed to convert $\Lambda_s^{\beta}$ in the form of matrix in \eqref{eq:inv_zero}. From this new matrix, the invariant zeros of the corresponding system are found.

To satisfy Lemma \ref{lem:condtns_complexpertgivens}, we require that $Q_2^\alpha\neq 0$ and rank($\Lambda_s^\beta$) $<n+p$. Since in both cases above, the set of $s\in\mathbb{C}$ that satisfies both conditions is finite, it is easy to compute $\|\Delta_s^*\|$ for all $s$ in this finite set, and obtain the minimum among these.

Further, from the above cases, and since invariant zeros cover either the whole complex plane or are countably finite, $s\in\mathbb{C}$ that satisfies Lemma \ref{lem:condtns_complexpertgivens} is either finite or cover the entire complex plane.\hfill$\square$
\end{remark}


\begin{figure}[htp]
    \vspace*{1mm}
    \centering
    \includegraphics[trim=2cm 0 0cm 0,clip,width=9.2cm]{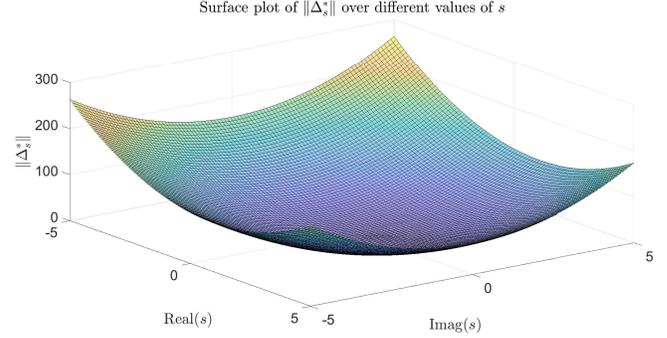}
    \caption[Surface Plot of $\|\Delta_s^*\|$ in Example \ref{ex:2}]{Surface plot of $\|\Delta_s^*\|$.  $s\in\mathbb{C}$ lies in the square area formed by the vertices $s=\pm 5\pm 5j$. The minimum of $\|\Delta_s^*\|=0.2086$ occurs at $s=0.8297+0.5583j$.}
    \label{fig:Example_Pert}
\end{figure}
\textbf{Example \ref{ex:2} (Contd.)} We consider the case where $A_0(1,1)$ and $A_0(3,1)$ are perturbed, which translates to perturbing the elements $A(1,1)$ and $A(1,3)$. We have
{\small\begin{align*}
\Lambda_s^\beta &= \begin{bmatrix}-0.69 & 1.62-sj & -0.21 & 0.71\\
-2.08 & 0.63 & 0.14-sj & 0.61\\
-1.23 & 1.02 & -0.66 & 1.33 \\ -0.26 & 2.51 & 1.13 & -2.89\end{bmatrix}.
\end{align*}}

For this case, we have that $\Lambda_s^\beta$ has rank equal to $n+p=4$ for all but finite $s\in\mathbb{C}$. Hence, the conditions in Lemma \ref{lem:condtns_complexpertgivens} are satisfied for only finite $s\in\mathbb{C}$. Therefore, we transform $\Lambda_s^\beta$, in order to find $s$ for which the rank is reduced. 

For this we perform row exchanges on $\Lambda_s^{\beta^H}$, such that the first row of $\Lambda_s^{\beta^H}$ becomes the third row, and the second and third rows of $\Lambda_s^{\beta^H}$ become the first and second rows, respectively:
{\small\begin{align*}
&&\text{rank}(\Lambda_s^{\beta^H})&<n+p\\
&\iff&\!\!\!\!\!\!\! \text{rank}\left(\begin{bmatrix}
1.62+\bar{s}j & 0.63 & 1.02 & 2.51\\
-0.21 & 0.14+\bar{s}j & -0.66 & 1.13 \\
-0.69 & -2.08 & -1.23 & -0.26\\
0.71 & 0.61 & 1.33 & -2.89
\end{bmatrix}\right)&<n+p,
\end{align*}}

\noindent where $\bar{s}$ is the conjugate of $s$. The set of $s\in\mathbb{C}$ that solves the above rank-reduction problem, can be calculated by finding the set of invariant zeros of the corresponding system with matrices:
{\small\begin{align*}
A&=\begin{bmatrix}1.62&0.63\\-0.21&0.14\end{bmatrix},B=\begin{bmatrix}1.02&2.51\\-0.66&1.13\end{bmatrix},\\
C&=\begin{bmatrix}-0.69&-2.08\\0.71&0.61\end{bmatrix},D=\begin{bmatrix}-1.23&-0.26\\1.33&-2.89\end{bmatrix}.
\end{align*}}

Accordingly, we get the invariant zeros as $s=0.8108\pm 0.5367j$. For these values of $s$, we have that rank$(\Lambda_s^{\beta^H})=3(<n+p=4)$. 
By using \eqref{eq:pertopt_prob1.9}, we get $\|\Delta_s\|=0.2097$ (rounded) for both these values of $s$. The corresponding perturbation (by \eqref{eq:pertopt_prob1.sbp1.1}) is found to be such that for the modified system, $\Delta(1,1)=-0.0270$ and $\Delta(1,3)=-0.2079$ (rounded), which has norm of $0.2097$ (rounded). The perturbed system $\Gamma^P$ has $\mathcal{V}(\Gamma^P)\neq\{0\}$.

A similar approach may be employed when $Q_2^\alpha=0$ for all but finite $s\in\mathbb{C}$.\hfill$\square$
\subsection{Computation Time}
On a 12th Gen Intel Core i9-12900K processor with 64 GB RAM, the time taken by Algorithm \ref{alg:IV.2.2} to compute the structured perturbation in Example \ref{ex:2} is 23 seconds (for $\Theta_0$ having discrete values of $\theta$ in steps of 0.01 and the algorithm stops when the number of elements in $\Theta_k$ is less than or equal to 1). This algorithm requires a single run on a single initial point, and hence is fast in finding the optimal solution. Further, it converges to the global minimum with high degree of confidence.
\section{Problem 2: Affine Perturbation}\label{sec:4.2}
We consider Problem \eqref{eq:pertopt_prob2.1}. A similar problem for the case of controllability was considered in \cite{RankRelaxation_arXiv_2022}. We modify the algorithm in \cite{RankRelaxation_arXiv_2022} to solve \eqref{eq:pertopt_prob2.1}.
The approach is as follows. Let the perturbations to $A$, $B$, $C$, $D$,  be $\Delta_A$, $\Delta_B$, $\Delta_C$, $\Delta_D$, respectively. Further, let $A_{\Delta}=A-\Delta_A$, $B_{\Delta}=B-\Delta_B$, $C_{\Delta}=C-\Delta_C$ and $D_{\Delta}=D-\Delta_D$, and let $s=\lambda+\mu j$.
By Lemma 1 of \cite{RankRelaxation_arXiv_2022}, the above problem is equivalent to
{\footnotesize\begin{equation}\label{eq:pertopt_prob2.4}
\begin{aligned}
&\arginf\limits_{x,\lambda\in\mathbb{R},\mu\in\mathbb{R},\Delta}&& \left\|\sum\limits_{i=1}^lE_i\Delta G_i \right \|\\\
    &\qquad\:s.t. && \underbrace{\begin{bmatrix}A_\Delta-\mu I & B_\Delta & -\lambda I & 0\\
     C_\Delta & D_\Delta & 0 & 0\\
     \lambda I & 0 & A_\Delta-\mu I & B_\Delta\\
     0 & 0 & C_\Delta & D_\Delta\end{bmatrix}}_{\tilde{\Lambda}_\Delta}x=0,
\end{aligned}
\end{equation}}

\noindent where $x\in\mathbb{R}^{(2(n+m))\times (2(n+p))}\backslash\{0\}$. 

The solution to Problem \eqref{eq:pertopt_prob2.4} is obtained by using Algorithm 2 of \cite{RankRelaxation_arXiv_2022}, which we summarize in Algorithm \ref{alg:IV.3}. We briefly discuss the algorithm below.

Let $Z=\tilde{\Lambda}_\Delta$. The constraint in Problem \eqref{eq:pertopt_prob2.4} (that is, $Z$ loses rank) is equivalent to having the last singular value of $Z$ equal to zero. The last singular value of $Z$ is obtained as $\|Z\|_*-\|Z\|_{F_{2(n+p)-1}}$, where $\|\cdot\|_{*}$ and $\|\cdot\|_{F_{r}}$ denote the nuclear norm (sum of all singular values) and the Ky Fan $r$-norm (sum of $r$ largest singular values), respectively. We penalize this constraint with a regularization parameter $\zeta>0$, and solve it iteratively by linearizing $\|Z\|_{F_{2(n+p)-1}}$ at $Z^{(k)}$. By linearizing, at the $(k+1)^{th}$ iteration we can replace $\|Z\|_{r_{2(n+p)-1}}$ by $\text{trace}(U_1^{(k)^T}ZV_1^{(k)})$, where $U_1^{(k)}$ and $V_1^{(k)}$ are matrices whose columns are the left and right singular vectors, respectively, of the corresponding largest $2(n+p)-1$ singular values of $Z^{(k)}$ and trace($\cdot$) is the sum of diagonal elements. Hence, Problem \eqref{eq:pertopt_prob2.4} is solved by solving the following convex problem at each iteration (Steps 6 and 7 of Algorithm \ref{alg:IV.3}):
{\footnotesize\begin{equation}\label{eq:pertopt_prob2.5}
\begin{aligned}
&\arginf\limits_{\lambda,\mu,\Delta,Z}&&\left\|\sum\limits_{i=1}^lE_i\Delta G_i\right\|+\zeta \left(\|Z\|_*-\text{trace}(U_1^{(k)^T}ZV_1^{(k)})\right)\\\
    &\qquad\:s.t. &&  Z=\tilde{\Lambda}_\Delta.
\end{aligned}
\end{equation}}

From \eqref{eq:pertopt_prob2.5}, we compute $Z^{(k+1)}$, and then update $U^{(k+1)}$ and $V^{(k+1)}$ using $Z^{(k+1)}$. The solution converges when  $F(Z^{(k)})=\left\|\sum\limits_{i=1}^lE_i\Delta^{(k)} G_i \right\|+\zeta \left(\|Z^{(k)}\|_*-\|Z^{(k)}\|_{F_{2(n+p)-1}}\right)$ converges, based on a specified threshold $\tau$ (Steps 4-6).

\begin{remark}
    To compute the global minimum, Algorithm \ref{alg:IV.3} takes more time to converge and requires multiple runs with different initial points. Hence, in solving the structured problem (that is, Problem \eqref{eq:pertopt_prob1.3}), Algorithm \ref{alg:IV.2.2} is preferred to Algorithm \ref{alg:IV.3}, since it converges to the global minimum much faster.\hfill$\square$
\end{remark}
\begin{algorithm}
\caption[Affine perturbation]{Affine perturbation}\label{alg:IV.3}
\textbf{Input:} Matrices $A$, $B$, $C$, $D$, $E_i$, $G_i\:\forall\: i\in\{1,2,\cdots,l\}$.\\
\textbf{Output:} $s,\|\sum\limits_{i=1}^lE_i\Delta^* G_i\|$ and $\Delta^*$ which solves \eqref{eq:pertopt_prob2.1}.\\~\\
\textbf{Step 1:} Let $k=0$. Let $\hat{\Delta}_s^{(0)}$, $\hat{\lambda}^{(0)}$ and $\hat{\mu}^{(0)}$ be chosen randomly.\\
\textbf{Step 2:} Compute $\hat{Z}$ using $\hat{\Delta}_s^{(0)}$, $\hat{\lambda}^{(0)}$ and $\hat{\mu}^{(0)}$. From singular value decomposition of $\hat{Z}$, find $\hat{U}_1^{(0)}$ and $\hat{V}_1^{(0)}$.\\
\textbf{Step 3:} Compute initial points $\Delta^{(0)},\lambda^{(0)},\mu^{(0)},Z^{(0)}$ by solving Problem \eqref{eq:pertopt_prob2.5} with its objective function changed to $\|Z\|_*-\text{trace}(\hat{U}_1^{(0)^T}Z\hat{V}_1^{(0)})$. Let $\tau>0$ be a threshold value.\\
\textbf{Step 4:} Let $F(Z^{(k)})=\|\sum\limits_{i=1}^lE_i\Delta^{(k)} G_i\|+\zeta(\|Z^{(k)}\|_*-\|Z^{(k)}\|_{F_{2(n+p)-1}})$.\\ 
\textbf{Step 5:} If $k=0$ or $|F(Z^{(k)})-F(Z^{(k-1)})|>\tau$ (for $k>0$), do Steps 6 to 8.\\
\textbf{Step 6:} Compute $U_1^{(k)}$ and $V_1^{(k)}$ from singular value decomposition of $Z^{(k)}$.\\
\textbf{Step 7:} Let $\zeta=\frac{\|\sum\limits_{i=1}^lE_i\Delta^{(0)} G_i\|}{\epsilon}$, for some $\epsilon\ll 1$. Compute $\Delta^{(k+1)},\lambda^{(k+1)},\mu^{(k+1)},Z^{(k+1)}$ by solving \eqref{eq:pertopt_prob2.5}.\\
\textbf{Step 8:} Update $k$ to $k+1$. Go to Step 4.\\
\textbf{Step 9:} Stop when convergence is reached.
\end{algorithm}
We can significantly improve the speed and accuracy of Algorithm \ref{alg:IV.3} by choosing the global solutions to the structured problem, Problem \eqref{eq:pertopt_prob1.3}, as the initial points of Algorithm \ref{alg:IV.3} (since the global minimum-norm perturbation and $s$ are generally close for the affine and structured cases).\par

\textbf{Example \ref{ex:2} (Contd.)} We consider the case where only $A_o(1,1)$, $A_o(1,3)$ and $A_o(3,1)$ are allowed to be perturbed in the system in Example 1 of \cite{Transmissionzerosradius_IFAC_2008}. 
As before, we consider the modified system \eqref{sys:ex2_modfd}. We use Algorithm \ref{alg:IV.3} to find the minimum perturbation. Earlier, we specified the global minimum perturbation for two cases considered in Example \ref{ex:2} in the previous section, that is, Case 1 where four elements $A(i,j)$, where $i,j\in\{1,3\}$, can be perturbed, and Case 2 where only two elements $A(1,1)$ and $A(1,3)$ can be perturbed. We choose these optimal values to initialize Algorithm \ref{alg:IV.3} in order to reach the global solution faster. Table \ref{tab:InitPts_Affine} contains these two initial points. Note that $\Delta_A(3,3)$ is set to $0$ for these initial points to follow the sparsity structure that only $A(1,1)$, $A(1,3)$ and $A(3,1)$ can be perturbed.
\begin{table}
\centering

\begin{tabular}{|c | c|} 
 \hline
 $\Delta_A$ (rounded) & $s$\\ [0.5ex] 
 \hline\hline
 $\begin{bmatrix}-0.0341 & 0 & -0.2048\\
  0 & 0 & 0\\
0.0682 & 0 & 0\end{bmatrix}$ & $0.8297+0.5583j$\\ 
 \hline
 $\begin{bmatrix}-0.0270 & 0 & -0.2079\\
  0 & 0 & 0\\
0 & 0 & 0\end{bmatrix}$ & $0.8108+0.5367j$\\ 
 \hline
\end{tabular}
\caption{Initial Points for Affine Perturbation}
\label{tab:InitPts_Affine}
\end{table}
With these initial points and  $\tau=10^{-8}$, $\epsilon=10^{-4}$. we get the optimal value of $\Delta_A$ (denoted as $\Delta_A^*$) at $s^*=0.8021 + 0.4948j$ for the modified system as (rounded to 4 decimal points):
{\small\begin{align*}
\Delta_A^*=\begin{bmatrix}-0.0096&-0.2079\\-0.1680&0\end{bmatrix},
\end{align*}}

\noindent and hence $\|E\Delta^*G\|=\|\Delta_A^*\|=0.2086$ (rounded). Note that this norm is very close to the global minimum in the structured case. Hence, this local minimum well approximates the global minimum for the affine case.\hfill$\square$

\subsection{Computation Time}
From \cite{RankRelaxation_arXiv_2022}, the time complexity of Algorithm \ref{alg:IV.3} is upper bounded by {$O((n+m)(n+p))$}. In our experiments on a 12th Gen Intel Core i9-12900K processor with 64 GB RAM, with $\tau=10^{-5},\epsilon=10^{-4}$ and the initial points in Table \ref{tab:InitPts_Affine}, the execution time is approximately 1 second but the solution has $\|\Delta^*\|=0.2111$, which is not accurate. To improve accuracy, we decreased the tolerance level to $\tau=10^{-8}$. With this, the algorithm converges to $\|\Delta^*\|=0.2086$, which is closer to the global solution. However, it takes approximately 45 minutes to converge. In general, we need multiple algorithm runs with different initial points to find a local minima that is close to the global minima. This increases the execution time. However, in our experiments, we needed only one run to achieve approximately global solution. This is because the initialization points in Table \ref{tab:InitPts_Affine} are global solutions of the structured case. These initial points guides the algorithm quickly to the solution.


\section{Conclusion}
In this paper, we proposed algorithms to convert a non-opaque system to an opaque one by minimally perturbing the system matrices. This enables the system operator to incorporate privacy (opacity) in the system. Future directions of work include assessing the minimum sparse perturbation to ensure other forms of opacity, such as current-state, $K$-step, etc., improving the algorithm on finding minimum structured perturbations, such that approximations are not required, developing algorithms that find the global minimum for affine perturbations with greater accuracy.
\section*{Appendix}
\subsection{Proof of $\|\Delta^r\|=\|E\Delta G\|$}\label{app:normeq}
We consider the two fundamental operations by which $\Delta^r$ is formed from $E\Delta G$ - removing all-zero rows and removing all-zero columns. In the following, we show that the norm of $E\Delta G$ does not change when either of these operations are performed, thus showing that $\|\Delta^r\|=\|E\Delta G\|$.\\
    \textbf{Removing zero row of $E\Delta G$:}\par
    We have by definition, $\|E\Delta G\|=\max\limits_{x\neq 0}\frac{\|E\Delta G x\|_2}{\|x\|_2}$.

    Since in this case, it holds that $\|E\Delta G x\|_2=\|\Delta^rx\|_2$, we have
    \begin{align*}
    \|E\Delta G\|=\max\limits_{x\neq 0}\frac{\|E\Delta G x\|_2}{\|x\|_2}=\max\limits_{x\neq 0}\frac{\|\Delta^rx\|_2}{\|x\|_2}=\|\Delta^r\|.
    \end{align*}
    
   \noindent \textbf{Removing zero column of $E\Delta G$:}\par
    We have by definition, $\|E\Delta G\|=\max\limits_{x\neq 0}\frac{\|E\Delta G x\|_2}{\|x\|_2}$.

    Let $y$ be the sub-vector of $x$ such that the elements of $x$ corresponding to the zero columns of $E\Delta G$ are not included. It is evident that $\|E\Delta G x\|_2=\|\Delta^ry\|_2$.\par 
    Therefore,
    \begin{align*}
    \max\limits_{x\neq 0}\frac{\|E\Delta G x\|_2}{\|x\|_2}=\max\limits_{x\neq 0}\frac{\|\Delta^ry\|_2}{\|x\|_2}.
    \end{align*}
    
    The maximum of the above expression is achieved when the elements of $x$ corresponding to the zero columns of $\Delta$ are set to 0. In this case, we have $\|x\|_2=\|y\|_2$, and hence,
    \begin{align*}
    \max\limits_{x\neq 0}\frac{\|\Delta^ry\|_2}{\|x\|_2}=\max\limits_{x\neq 0}\frac{\|\Delta^ry\|_2}{\|y\|_2},
    \end{align*}
    
    \noindent which is equal to
    \begin{align*}
    \max\limits_{y\neq 0}\frac{\|\Delta^ry\|_2}{\|y\|_2}=\|\Delta^r\|.
    \end{align*}
\subsection{Analysis of Approximation For Structured Perturbation}\label{app:approxanalysis}
We consider the approximation for structured perturbation given by \eqref{eq:pertopt_prob_set_1}. For this approximation, we take an example for which the approximation fails for some values of $\epsilon$ that approximates $E$ and $G$ to $\breve{E}$ and $\breve{G}$, respectively. With this example, we show that for all practical control systems, this approximation is valid.

Consider the system with matrices given by:
{\small\begin{align*}
A=\begin{bmatrix}10^4\end{bmatrix},B=\begin{bmatrix}0&0\end{bmatrix},C=\begin{bmatrix}10^{-4}\\0\end{bmatrix},D=\begin{bmatrix}10^4&10^4\\0&10^{-4}\end{bmatrix}.
\end{align*}}

For this system, $\Lambda_s$ is given by:
{\small\begin{align*}
    \Lambda_s=\begin{bmatrix}10^4-s&0&0\\
    10^{-4}&10^4&10^4\\
    0&0&10^{-4}\end{bmatrix},
\end{align*}}

\noindent for which we perturb only $\Lambda_s(1,2)$. This example serves as one of the significant corner cases for which the approximation might fail for some values of $\epsilon$. 
We first consider $s=0$ to that the minimum-perturbation norm found by the approximation is accurate. After this, we show that this approximation is accurate for any $s\in\mathbb{C}$ as well. With this, we consider the following perturbation:
{\small\begin{align*}
    \Lambda_s-E\Delta G&=\begin{bmatrix}10^4&0&0\\
    10^{-4}&10^4&10^4\\
    0&0&10^{-4}\end{bmatrix}-\\&\begin{bmatrix}1&0&0\\
    0&0&0\\
    0&0&0\end{bmatrix}\begin{bmatrix}\Delta(1,1)&\Delta(1,2)&\Delta(1,3)\\
    \Delta(2,1)&\Delta(2,2)&\Delta(2,3)\\
    \Delta(3,1)&\Delta(3,2)&\Delta(3,3)\end{bmatrix}\begin{bmatrix}0&0&0\\
    0&1&0\\
    0&0&0\end{bmatrix}\\
    &=\begin{bmatrix}10^4&-\Delta(1,2)&0\\
    10^{-4}&10^4&10^4\\
    0&0&10^{-4}\end{bmatrix}.
\end{align*}}

For this matrix, we use \eqref{eq:pertopt_prob1.9} and \eqref{eq:pertopt_prob1.sbp1.1} to get $\|\Delta_s^*\|=\|\Delta(1,2)\|=10^{12}$, for which
{\small\begin{align*}
\Delta_s^*=\begin{bmatrix}
    0&-10^{12}&0\\
    0&0&0\\
    0&0&0
\end{bmatrix}.\numberthis\label{eq:approx_Deltastar}
\end{align*}}

Due to this, we get $\Lambda_s-E\Delta_s^*G$ as
{\small\begin{align*}
    \Lambda_s-E\Delta_s^* G&=\begin{bmatrix}10^4&10^{12}&0\\
    10^{-4}&10^4&10^4\\
    0&0&10^{-4}\end{bmatrix}.
\end{align*}}

Since the first and second columns of $\Lambda_s-E\Delta_s^*G$ are dependent, we have rank$(\Lambda_s-E\Delta_s^*G)=2<n+p\:(=3)$.

Let us consider the approximation given by \eqref{eq:pertopt_prob_set_1}. We first show that if $\epsilon$ is not small enough, the approximation is inaccurate, that is, the solution to Problem \eqref{eq:pertopt_prob_set_1} is far from the solution to Problem \eqref{eq:pertopt_prob1.6}. After this, we show that for smaller values of $\epsilon$, the approximation is accurate, that is, the solutions to these two problems are very close.

We add $\epsilon$ to the diagonal elements of $E$ and $G$, such that $\breve{E}$ and $\breve{G}$ are:
{\small\begin{align*}
    \breve{E}=\begin{bmatrix}
    1&0&0\\
    0&\epsilon&0\\
    0&0&\epsilon
    \end{bmatrix},\breve{G}=\begin{bmatrix}
    \epsilon&0&0\\
    0&1&0\\
    0&0&\epsilon
    \end{bmatrix}.
\end{align*}}

Let the perturbation in this case be denoted by $\breve{\Delta}$. Therefore, we have:
{\small\begin{align*}\breve{E}\breve{\Delta}\breve{G}=
\begin{bmatrix}\epsilon\breve{\Delta}(1,1)&\breve{\Delta}(1,2)&\epsilon\breve{\Delta}(1,3)\\
    \epsilon^2\breve{\Delta}(2,1)&\epsilon\breve{\Delta}(2,2)&\epsilon^2\breve{\Delta}(2,3)\\
    \epsilon^2\breve{\Delta}(3,1)&\epsilon\breve{\Delta}(3,2)&\epsilon^2\breve{\Delta}(3,3)\end{bmatrix}.
\end{align*}}

In this example, if $\epsilon$ is not small enough, then $\breve{\Delta}(3,2)$ will be taken into consideration during the perturbation, such that the rank reduces. For instance, if we consider $\epsilon=10^{-4}$, we get:
{\small\begin{align*}\breve{E}\breve{\Delta}\breve{G}=
\begin{bmatrix}10^{-4}\breve{\Delta}(1,1)&\breve{\Delta}(1,2)&10^{-4}\breve{\Delta}(1,3)\\
    10^{-8}\breve{\Delta}(2,1)&10^{-4}\breve{\Delta}(2,2)&10^{-8}\breve{\Delta}(2,3)\\
    10^{-8}\breve{\Delta}(3,1)&10^{-4}\breve{\Delta}(3,2)&10^{-8}\breve{\Delta}(3,3)\end{bmatrix}.
\end{align*}}

Consequently, we have
{\small\begin{align*}
    &\Lambda_s-\breve{E}\breve{\Delta}\breve{G}=\\
    &\begin{bmatrix}10^4-10^{-4}\breve{\Delta}(1,1)&\!\!\!-\breve{\Delta}(1,2)&\!\!\!-10^{-4}\breve{\Delta}(1,3)\\
    10^{-4}-10^{-8}\breve{\Delta}(2,1)&\!\!\!10^4-10^{-4}\breve{\Delta}(2,2)&\!\!\!10^4-10^{-8}\breve{\Delta}(2,3)\\
    -10^{-8}\breve{\Delta}(3,1)&\!\!\!-10^{-4}\breve{\Delta}(3,2)&\!\!\!10^{-4}-10^{-8}\breve{\Delta}(3,3)\end{bmatrix}.
\end{align*}}

By using the approximation, we obtain 
{\small\begin{align*}
    \breve{\Delta}_s^*=\begin{bmatrix}0&0&0\\
    0&0&0\\
    0&-1&0
    \end{bmatrix},\numberthis\label{eq:solnapprox1}
\end{align*}}

\noindent such that 
{\small\begin{align*}
    \Lambda_s-\breve{E}\breve{\Delta}_s^*\breve{G}=\begin{bmatrix}10^4&\!\!\!0&\!\!\!0\\
    10^{-4}&\!\!\!10^4&\!\!\!10^4\\
    0&\!\!\!10^{-4}&\!\!\!10^{-4}\end{bmatrix}.
\end{align*}}

Since the second and third column vectors of $\Lambda_s-\breve{E}\breve{\Delta}_s^*\breve{G}$ are dependent, we have rank$(\Lambda_s-\breve{E}\breve{\Delta}_s^*\breve{G})=2<n+p\:(=3)$. Further, the norm of $\breve{\Delta}_s^*$ is given by $\|\breve{\Delta}_s^*\|=1$. This is much smaller than the norm of the perturbation computed without approximation, that is, $\Delta_s^*$ in \eqref{eq:approx_Deltastar}, whose norm is $\|\Delta_s^*\|=10^{12}$. Hence, for $\epsilon=10^{-4}$, we get an inaccurate minimum-norm perturbation. This is because the approximation with $\epsilon$ relaxes the structure enforced by $\Lambda_s-E\Delta G$ (for which $\epsilon=0$). Consequently, elements in $\Lambda_s$ that cannot be perturbed by $E\Delta G$ are allowed to be perturbed by the approximation (for which $\epsilon\neq 0$). This relaxation has a greater effect when $\epsilon$ has a higher value. Hence, we require $\epsilon$ to be lower such that this relaxation has negligible effect on the perturbation structure.

Let $\epsilon=10^{-20}$. For this $\epsilon$, the minimum-norm perturbation that solves \eqref{eq:pertopt_prob_set_1} is given by:
{\small\begin{align*}
    \breve{\Delta}_s^*=\begin{bmatrix}0&-10^{12}&0\\
    0&0&0\\
    0&0&0
    \end{bmatrix},
\end{align*}}

\noindent for which $\|\breve{\Delta}_s^*\|=10^{12}$, which is equal to the minimum perturbation norm $\|\Delta_s^*\|$ given by \eqref{eq:approx_Deltastar} (without approximation). Further $\breve{\Delta}_s^*$ reduces the rank of $\Lambda_s$ as:
{\small\begin{align*}
    &\Lambda_s-\breve{E}\breve{\Delta}_s^*\breve{G}=
    \\&\begin{bmatrix}10^4-10^{-20}\breve{\Delta}_s^*(1,1)&\!\!\!-\breve{\Delta}_s^*(1,2)&\!\!\!-10^{-20}\breve{\Delta}_s^*(1,3)\\
    10^{-4}-10^{-40}\breve{\Delta}_s^*(2,1)&\!\!\!10^4-10^{-20}\breve{\Delta}_s^*(2,2)&\!\!\!10^4-10^{-40}\breve{\Delta}_s^*(2,3)\\
    -10^{-40}\breve{\Delta}_s^*(3,1)&\!\!\!-10^{-20}\breve{\Delta}_s^*(3,2)&\!\!\!10^{-4}-10^{-40}\breve{\Delta}_s^*(3,3)\end{bmatrix}\\
    &=\begin{bmatrix}10^4&\!\!\!10^{12}&\!\!\!0\\
    10^{-4}&\!\!\!10^4&\!\!\!10^4\\
    0&\!\!\!0&\!\!\!10^{-4}\end{bmatrix},
\end{align*}}

\noindent has the first and second column vectors as dependent.

To compare with $\epsilon=10^4$, let us consider the perturbation $\breve{\Delta}$ with same structure as $\breve{\Delta}_s^*$ in \eqref{eq:solnapprox1}:
{\small\begin{align*}
\breve{\Delta}=\begin{bmatrix}0&0&0\\
    0&0&0\\
    0&-10^{16}&0
    \end{bmatrix}.
\end{align*}}

This $\breve{\Delta}$ reduces the rank of $\Lambda_s$ since
{\small\begin{align*}
\Lambda_s-\breve{E}\breve{\Delta}\breve{G}=\begin{bmatrix}10^4&\!\!\!0&\!\!\!0\\
    10^{-4}&\!\!\!10^4&\!\!\!10^4\\
    0&\!\!\!10^{-4}&\!\!\!10^{-4}\end{bmatrix},
\end{align*}}

\noindent such that the second and third column vectors of $(\Lambda_s-\breve{E}\breve{\Delta}\breve{G})$ are dependent. However, note that $\|\breve{\Delta}\|=10^{16}$, which is greater than $\|\Delta_s^*\|=10^{12}$. Thus, $\breve{\Delta}$ is not the minimum-norm perturbation.


Hence, by lowering $\epsilon$ to $\epsilon=10^{-20}$, the approximation gives perturbations closer to the optimal perturbation $\Delta_s^*$ in \eqref{eq:approx_Deltastar}. This is because a lower value of $\epsilon$ constrains $\breve{E}\breve{\Delta}\breve{G}$ to only perturb $\Lambda_s(1,2)$.

From the above, we see that for this system we require $\epsilon<10^{-16}$ for accuracy. The drawback with a very low $\epsilon$ is that it demands a higher numerical precision, thereby increasing the computation time. 
We tabulate the minimum perturbation norm obtained through approximation for different values of $\epsilon$ in Table \ref{tab:Deltas_epsilon}.

\begin{table}
\centering

\begin{tabular}{|c | c| c|} 
 \hline
 $\epsilon$ & $\|\Delta_s^*\|$ by approximation (round)\\ [0.5ex] 
 \hline\hline
 $10^{-4}$ & 1\\ 
 \hline
  $10^{-8}$& $10^{4}$\\ 
   
 \hline
 $10^{-12}$& $10^{8}$\\ 
 \hline
 $10^{-16}$& $7.0711\times 10^{11}$\\ 
 \hline
 $10^{-20}$& $10^{12}$\\ 
 \hline
 $10^{-50}$& $10^{12}$\\ 
 \hline
 
 $10^{-100}$& $10^{12}$\\ 
 \hline
 $10^{-200}$& $10^{12}$\\ 
 \hline
 $10^{-323}$& $10^{12}$\\ 
 \hline
 
\end{tabular}
\caption{$\|\Delta_s^*\|$ by approximation for $\delta_M=10^4$ and different values of $\epsilon$}
\label{tab:Deltas_epsilon}
\end{table}

For real-world control systems, the elements of the system matrices $A,B,C,D$ are bounded above and below by finite values. Therefore, we consider the approximation for different bounds on the element values. Let $\delta_M$ and $\frac{1}{\delta_M}$ denote the maximum and minimum values, respectively, of the elements of $\Delta$ and $\breve{\Delta}$.

We analyze $\Lambda_s$ for these bounds. Hence, we have for $s=0$,
{\small\begin{align*}
\Lambda_s-E\Delta G&=\begin{bmatrix}\delta_M&-\Delta(1,2)&0\\
    \delta_M^{-1}&\delta_M&\delta_M\\
    0&0&\delta_M^{-1}\end{bmatrix},\\
    \Lambda_s-\breve{E}\breve{\Delta} \breve{G}&=\\&\!\!\!\!\!\!\!\!\!\begin{bmatrix}\delta_M-\epsilon\breve{\Delta}(1,1)&-\breve{\Delta}(1,2)&-\epsilon\breve{\Delta}(1,3)\\
    \delta_M^{-1}-\epsilon^2\breve{\Delta}(2,1)&\delta_M-\epsilon\breve{\Delta}(2,2)&\delta_M-\epsilon^2\breve{\Delta}(2,3)\\
    -\epsilon^2\breve{\Delta}(3,1)&-\epsilon\breve{\Delta}(3,2)&\delta_M^{-1}-\epsilon^2\breve{\Delta}(3,3)\end{bmatrix}\!.
\end{align*}}

For this, we have that when $-\Delta(1,2)=\delta_M^3$ and all other elements of $\Delta$ are 0, then $\Lambda_s$ loses rank. Also, the minimum perturbation by the approximation such that $\Lambda_s$ loses rank is such that $-\epsilon\breve{\Delta}(3,2)=\delta_M^{-1}$ and all other elements of $\breve{\Delta}_s$ are 0. In this case, for $\|\Delta\|$ to be equal to $\|\breve{\Delta}\|$, then it is required that
{\small\begin{align*}
    &&|\Delta(1,2)|&=|\breve{\Delta}(3,2)|\\
    &\implies& \delta_M^3&=\frac{\delta_M^{-1}}{\epsilon}\\
    &\implies& \epsilon&=\delta_M^{-4}.
\end{align*}}

For each $\delta_M$, we can find a suitable value of $\epsilon$ from the above equation and from the data in Table \ref{tab:Deltas_epsilon}. From this table, we have that for $\frac{1}{\delta_{M}}\leq|\Delta(i,j)|\leq \delta_{M}$ (where $i\in\{1,2,\cdots,n+m\},j\in\{1,2,\cdots,n+p\}$), the value of $\epsilon=\frac{1}{10^4\delta_{M}^4}$ provides sufficient accuracy.

Next, we consider the maximum upper bound $\delta_M$ of the elements of $\Lambda_s$ such that the approximation $\|\breve{\Delta}_s^*\|$ is valid. This is done empirically for different values of $\delta_M$ and the value of $\epsilon=\frac{1}{10^4\delta_{M}^4}$. The corresponding values of $\|\Delta_s^*\|$ (by \eqref{eq:pertopt_prob1.9}) and $\|\breve{\Delta}_s^*\|$ (by \eqref{eq:approx_norm}) are tabulated in \footnote{The values of $\phi$ in Table \ref{tab:Deltas_deltaM} are due to the fact that for these values of $\delta_M$, $\theta$ and $\|\Delta_{s_{k-1}}^*\|$, there do not exist $|s|$ for which $\|\Delta_{|s|e^{j\theta}}^*\|=\|\Delta_{s_{k-1}}^*\|$.}{Table \ref{tab:Deltas_deltaM}}. It is evident from this data that the approximation is valid at least up to $\delta_M=10^{75}$. This value of $\delta_M$ corresponds to all practical control systems, allowing our approximation to hold valid in practice.
\begin{table}
\centering

\begin{tabular}{|c |c| c| c|c|c|} 
 \hline
 $\delta_M$&$\epsilon$ & $\|\Delta_s^*\|$(exact) & $\|\breve{\Delta}_s^*\|$ (round) & $\frac{|\text{Error}|}{\|\Delta_s^*\|}\%$\\ [0.5ex] 
 \hline\hline
 $10^{1}$& $10^{-8}$ & $10^3$ & $10^3$ & $5.0010\times 10^{-7}$\\ 
   \hline
 $10^{5}$ & $10^{-24}$ & $10^{15}$& $10^{15}$ & $5\times 10^{-7}$\\
 \hline
 $10^{10}$ & $10^{-44}$ & $10^{30}$& $10^{30}$ & $5\times 10^{-7}$\\ 

 \hline
 $10^{25}$ & $10^{-104}$ & $10^{75}$& $10^{75}$ & $5\times 10^{-7}$\\
 \hline
 $10^{50}$ & $10^{-204}$ & $10^{150}$& $10^{150}$ & $5\times 10^{-7}$\\
 \hline
 
 $10^{75}$ & $10^{-304}$ & $10^{225}$& $10^{225}$ & $5\times 10^{-7}$\\
 \hline
\end{tabular}
\caption{$\|\Delta_s^*\|$ by approximation for different values of $\delta_M$ and $\epsilon$}
\label{tab:Deltas_deltaM}
\end{table}

Next, we show that this approximation holds true for systems with large dimensional matrices. For this, we consider the system whose $\Lambda_s$ is a block-diagonal matrix of size $999\times 999$. Each block in the diagonal is the matrix in the previous example, with $\delta_M=10^{10}$ ($s=0$):
{\small\begin{align*}
    \begin{bmatrix}10^{10}&0&0\\
    10^{-10}&10^{10}&10^{10}\\
    0&0&10^{-10}\end{bmatrix}.
\end{align*}}

Here as well, we perturb only the element $\Lambda_s(1,2)$, and $\epsilon=10^{-42}$. From the computation, we have $\|\Delta_s^*\|=10^{30}$ (exact) and $\|\breve{\Delta}_s^*\|=10^{30}$ (round). For this case, we have $\frac{|\text{Error}|}{\|\Delta_s^*\|}\%=5\times 10^{-7}$. This shows the robustness of the approximation to large dimensional system matrices.

Next, we show that with this approximation, given $\|\Delta_{s_{k-1}}^*\|$ (for some $k>0$) and $\theta\in[0,2\pi]$, the procedure detailed in Section \ref{sec:s_forpert} outputs the exact values of the set of $|s|$ satisfying $\|\Delta_{|s|e^{j\theta}}^*\|=\|\Delta_{s_{k-1}}^*\|$. This also shows that the approximation is accurate for arbitrary $s\in\mathbb{C}$, and not just for $s=0$. The values of $|s|$ obtained for different values of $\delta_M$, $\theta$ and $\|\Delta_{s_{k-1}}^*\|$ are tabulated in Table \ref{tab:svals}. We have chosen $\epsilon=10^{-320}$. As shown by the zero values in the $|\text{Error}|$ column of this table, the computed values of $|s|$ give the exact perturbation norm $\|\Delta_{|s|e^{j\theta}}^*\|=\|\Delta_{s_{k-1}}^*\|$. This shows empirically that the computation of the set of $|s|$ is accurate for arbitrary $\theta\in[0,2\pi]$. We show this further by a Monte Carlo analysis below.

We perform Monte Carlo analysis over $10^4$ runs to demonstrate that the computed set of $|s|$ results in the exact perturbation norm $\|\Delta_{s_{k-1}}^*\|$. In each run, we choose a random value from the set $[0,2\pi]$ for $\theta$, and random values from the set $[0,10^3]$ for the elements of matrices $A\in\mathbb{R}^{3\times 3}$, $B\in\mathbb{R}^{3\times 2}$, $C\in\mathbb{R}^{2\times 3}$ and $D\in\mathbb{R}^{2\times 2}$. We choose $\epsilon=10^{-150}$ and $\|\Delta_{s_{k-1}}^*\|=10^{230}$. We constrain the perturbation $\Delta$ to only perturb $A(1,2)$. Hence, $E=\text{diag}(1,0,0,0,0)$ and $G=\text{diag}(0,1,0,0,0)$. The output of the analysis is as follows. For $100\%$ of the $10^4$ runs, this method computes $|s|$ for which $\|\Delta_{|s|e^{j\theta}}^*\|=\|\Delta_{s_{k-1}}^*\|=10^{230}$. This result demonstrates again the accuracy of the approximation.
\begin{table}
\centering

\begin{tabular}{|c |c| c| c|c|c|} 
 \hline
 \!$\delta_M$\!&\!$\theta$\! & \!$\|\Delta_{s_{k-1}}^*\|$\! & \!$|s|$ (largest, round) \!& \!$\|\Delta_{|s|e^{j\theta}}^*\|$\!& \!\!$|\text{Error}|$\!\!\\ [0.5ex] 
 \hline\hline
 $10$ & $0$ & $10^3$ & $20$ & $10^3$ & $0$ \\ \hline
$10$ & $0$ & $10^{230}$ & $10^{228}$ & $10^{230}$ & $0$ \\ \hline
$10$ & $\pi/4$ & $10^3$ & $7.4336 \times 10^{-19}$ & $10^3$ & $0$ \\ \hline
$10$ & $\pi/4$ & $10^{230}$ & $1.4142 \times 10^{225}$ & $10^{230}$ & $0$ \\ \hline
$10$ & $\pi/3$ & $10^3$ & $6.0747 \times 10^{-19}$ & $10^3$ & $0$ \\ \hline
$10$ & $\pi/3$ & $10^{230}$ & $1.1547 \times 10^{225}$ & $10^{230}$ & $0$ \\ \hline
$10$ & $\pi/2$ & $10^3$ & $5.2670 \times 10^{-19}$ & $10^3$ & $0$ \\ \hline
$10$ & $\pi/2$ & $10^{230}$ & $10^{225}$ & $10^{230}$ & $0$ \\ \hline
$10$ & $5\pi/6$ & $10^3$ & $1.0571 \times 10^{-18}$ & $10^3$ & $0$ \\ \hline
$10$ & $5\pi/6$ & $10^{230}$ & $2 \times 10^{225}$ & $10^{230}$ & $0$ \\ \hline
$10^{75}$ & $0$ & $10^3$ & $10^{75}$ & $10^3$ & $0$ \\ \hline
$10^{75}$ & $0$ & $10^{230}$ & $10^{80}$ & $10^{230}$ & $0$ \\ \hline
$10^{75}$ & $\pi/4$ & $10^3$ & $\phi$ & $-$ & $-$ \\ \hline
$10^{75}$ & $\pi/4$ & $10^{230}$ & $1.4142\times 10^{77}$ & $10^{230}$ & $0$ \\ \hline
$10^{75}$ & $\pi/3$ & $10^3$ & $\phi$ & $-$ & $-$ \\ \hline
$10^{75}$ & $\pi/3$ & $10^{230}$ & $1.1547\times 10^{77}$ & $10^{230}$ & $0$ \\ \hline
$10^{75}$ & $\pi/2$ & $10^3$ & $\phi$ & $-$ & $-$ \\ \hline
$10^{75}$ & $\pi/2$ & $10^{230}$ & $10^{77}$ & $10^{230}$ & $0$ \\ \hline
$10^{75}$ & $5\pi/6$ & $10^3$ & $\phi$ & $-$ & $-$ \\ \hline
$10^{75}$ & $5\pi/6$ & $10^{230}$ & $2\times 10^{77}$ & $10^{230}$ & $0$ \\ \hline
\end{tabular}
\caption{Values of $|s|$ for different values of $\Delta_{s_{k-1}}^*$ and $\theta$}
\label{tab:svals}
\end{table}


\end{document}